\documentclass[runningheads]{llncs}

\usepackage[bibliography=common]{apxproof}
\newtheoremrep{theorem}{Theorem}[section]
\newtheoremrep{lemma}{Lemma}[section]

\usepackage[noend]{algpseudocode}
\usepackage{algorithm}
\usepackage{xcolor}
\usepackage{bm}
\usepackage{enumitem}
\usepackage{listings}
\usepackage{bussproofs}
\usepackage{amsmath,amssymb}
\definecolor{vblue}{rgb}{.1,.15,.62}
\usepackage[colorlinks=true,citecolor=black,linkcolor=vblue,urlcolor=vblue]{hyperref}
\usepackage{tikz}
\usepackage{mathpartir}
\usepackage{xargs}
\usepackage{stmaryrd}
\usepackage{graphicx}
\usepackage{mathtools}
\usepackage{relsize}
\usepackage{xparse}
\usepackage{subcaption}
\usepackage{graphicx}

\usepackage{pgfplots}

\usepackage{units}
\newcommand{\m}{\unit{m}}
\newcommand{\s}{\unit{s}}
\newcommand{\kg}{\unit{kg}}
\newcommand{\Ne}{\unit{N}}
\newcommand{\ms}{\nicefrac{\m}{\s}}
\newcommand{\mss}{\nicefrac{\m}{\s^2}}

\usepackage{wrapfig}
\usepackage[capitalise]{cleveref}
\Crefname{section}{Sec.}{Secs.}
\Crefname{figure}{Fig.}{Figs.}
\Crefname{tabular}{Tab.}{Tabs.}
\Crefname{definition}{Def.}{Defs.}
\Crefname{theorem}{Theorem}{Theorems.}
\usepackage{xspace}

\definecolor{americanrose}{rgb}{1.0, 0.01, 0.24}
\definecolor{burgundy}{rgb}{0.5, 0.0, 0.13}
\definecolor{bostonuniversityred}{rgb}{0.8, 0.0, 0.0}
\newcommand{\dL}[1][]{\text{\upshape\textsf{d{\kern-0.05em}L}}\xspace}

  \newcommand{\setEnvelope}{\mathcal{E}}
  \newcommand{\formulaEnvelope}{E}

  \newcommand{\termA}{e}
  \newcommand{\termB}{g}
  \newcommand{\termFuncA}{f(x)}
  \newcommand{\programA}{\alpha}
  \newcommand{\programB}{\beta}
  \newcommand{\formulaA}{P}
  \newcommand{\formulaB}{Q}
  \newcommand{\progOde}{x'= \termFuncA \; \& \; \formulaB}
  \newcommand{\progOdeTime}{x'= f(x), t'=1 \; \& \; t \leq \constSamplingTime}
  \newcommand{\progOdeTimeControl}{x'= f(x,u), t'=1 \; \& \; t \leq \constSamplingTime}
  \newcommand{\constSamplingTime}{\Delta t}

  \newcommand{\setRCI}{\mathcal{S}}
  \newcommand{\setInitial}{\mathcal{X}_0}
  \newcommand{\setSafety}{\mathcal{X}}
  \newcommand{\setControl}{\mathcal{U}}

  \newcommand{\formulaRCI}{S(x,u)}
  \newcommand{\formulaInitial}{X_0(x)}
  \newcommand{\formulaSafety}{X(x)}

  \NewDocumentCommand{\setReach}{ O{t} O{\setInitial} O{\setEnvelope}}{%
    \mathcal{R}\bigl(#1,#2,#3)}
  \NewDocumentCommand{\setReachApp}{ O{t} O{\setInitial} O{\setEnvelope}}{\widehat{\mathcal{R}}\bigl(#1,#2,#3\bigr)}

  \newcommand{\define}{\mathrel{:=}}
  \newcommand{\matrixZonoA}{G}
  \newcommand{\matrixZonoB}{H}
  \newcommand{\vecZonoA}{c}
  \newcommand{\vecZonoB}{b}

  \newcommand{\zonotope}[2]{\bm{\langle} #1,\ #2 \bm{\rangle}}
  \newcommand{\vecZonoCoeff}{\lambda}
  \newcommand{\formulaZonoA}{\zonotope{\vecZonoA}{\matrixZonoA}}
  \newcommand{\formulaZonoB}{\zonotope{\vecZonoB}{\matrixZonoB}}
  \newcommand{\timeDerivative}[1]{\partial_t #1}

  \newcommand{\intA}{\mathbf{I}}
  \newcommand{\intB}{\mathbf{J}}
  \newcommand{\intRemainder}{\mathbf{R}}

  \newcommand{\setIntervals}{\mathbb{I}\mathbb{Q}}

  \newcommand{\opIntMid}[1]{\operatorname{mid}(#1)}
  \newcommand{\opIntRad}[1]{\operatorname{rad}(#1)}
  \newcommand{\opIntLower}[1]{\underline{#1}}
  \newcommand{\opIntUpper}[1]{\overline{#1}}
  \newcommand{\opIntEval}[3]{\left[#1\right]^{#2}_{#3}}

\NewDocumentCommand{\TM}{ m m m O{t}}{\mathrm{TM}_{#1,#2}\!\bigl(#3,#4\bigr)}
\NewDocumentCommand{\TMDerivative}{ m m m O{t}}{%
  \timeDerivative{\mathrm{TM}}_{#1,#2}\!\bigl(#3,#4\bigr)}

\makeatletter
\NewDocumentCommand{\axset}{m m}{%
  \expandafter\newcommand\csname axdisplay@#1\endcsname{#2}%
}
\newcommand{\axget}[1]{%
  \@ifundefined{axdisplay@#1}{#1}{\csname axdisplay@#1\endcsname}%
}
\NewDocumentCommand{\axtag}{m}{%
  \tag*{\color{gray}\axget{#1}}\label{ax:#1}%
}
\NewDocumentCommand{\axref}{m}{%
  \hyperref[ax:#1]{\axget{#1}}%
}
\newcommand{\ruleset}[2]{%
  \expandafter\newcommand\csname ruledisplay@#1\endcsname{#2}%
}
\newcommand{\ruleget}[1]{%
  \@ifundefined{ruledisplay@#1}{#1}{\csname ruledisplay@#1\endcsname}%
}

\newcommand{\Vast}{\bBigg@{6.5}}
\newcommand{\savedtarget}[2]{%
  \Hy@raisedlink{\hypertarget{#1}{}}%
  \protected@write\@mainaux{}{%
    \string\expandafter\string\gdef
    \string\csname\string\detokenize{#1}\string\endcsname{#2}%
  }%
}

\newcommand{\link}[1]{%
  \hyperlink{#1}{\csname #1\endcsname}%
}
\makeatother
\newcommand{\ruletag}[1]{%
  \phantomsection
  \LeftLabel{\textcolor{gray}{\ruleget{#1}}}%
  \label{rule:#1}%
}
\newcommand{\ruleref}[1]{%
  \hyperref[rule:#1]{\ruleget{#1}}%
}
\newcommand{\init}{\operatorname{init}}
\newcommand{\AND}{\; \& \;}
\newcommand{\R}{\mathbb{R}}
\newcommand{\Q}{\mathbb{Q}}
\newcommand{\vl}{\;|\;}

\newcommand{\sys}{\operatorname{sys}}
\newcommand{\ctrl}{\operatorname{ctrl}}
\newcommand{\plant}{\operatorname{plant}}

\NewDocumentCommand{\norm}{O{\,\cdot\,}}{\left\lVert#1\right\rVert}
\NewDocumentCommand{\normInf}{O{\,\cdot\,}}{\left\lVert#1\right\rVert_{\infty}}
\newcommand{\abs}[1]{\left|#1\right|}

\newcommand\blue[1]{\textcolor{vblue}{#1}}

\newcommandx{\argmax}[1][]{\underset{#1}{\mathrm{arg~max}}}
\newcommand{\uctrl}{u}

\newcommandx{\cinterval}[2][1=a,2=b, usedefault]{\left[#1,#2\right]}
\titlerunning{From Zonotopes to Proof Certificates}
\EnableBpAbbreviations
\begin{document}
\title{From Zonotopes to Proof Certificates:\\A Formal Pipeline for Safe Control Envelopes}
\author{Jonathan Hellwig\inst{1}\orcidID{0009-0009-5530-3256} \and
  Lukas Schäfer \inst{2}\orcidID{0000-0002-4335-9342} \and
Long Qian \inst{3}\orcidID{0000-0003-1567-3948} \and
André Platzer \inst{1}\orcidID{0000-0001-7238-5710} \and
  Matthias Althoff \inst{2}\orcidID{0000-0003-3733-842X}
}
\authorrunning{J. Hellwig et al.}
\institute{Karlsruhe Institute of Technology, Karlsruhe, Germany \\
  \email{\{jonathan.hellwig, platzer\}@kit.edu} \and
  Technical University of Munich, Garching, Germany \\
  \email{\{lukas.schaefer, althoff\}@tum.de} \and
  Carnegie Mellon University, Pittsburgh, United States of America \\
  \email{longq@andrew.cmu.edu}
}
\maketitle
\begin{abstract}
  Synthesizing controllers that enforce both safety and actuator constraints is a central challenge in the design of cyber-physical systems.
  State-of-the-art reachability methods based on zonotopes deliver impressive scalability, yet no zonotope reachability tool has been formally verified and the lack of end-to-end correctness undermines the confidence in their use for safety-critical systems.
  Although deductive verification with the hybrid system prover KeYmaera X could, in principle, resolve this assurance gap, the high-dimensional set representations required for realistic control envelopes overwhelm its reasoning based on quantifier elimination.
  To address this gap, we formalize how control-invariant sets serve as sound safety certificates.
  Building on that foundation, we develop a verification pipeline for control envelopes that unites scalability and formal rigor.
  First, we compute control envelopes with high-performance reachability algorithms.
  Second, we certify every intermediate result using provably correct logical principles.
  To accelerate this certification, we offload computationally intensive zonotope containment tasks to efficient numerical backends, which return compact witnesses that KeYmaera X validates rapidly.
  We show the practical utility of our approach through representative case studies.
\end{abstract}

\section{Introduction}
Autonomous vehicles, aircraft, and other cyber-physical systems increasingly demand advanced control algorithms while never violating stringent safety specifications.
Control-envelope synthesis has become a central problem at the intersection of formal verification and control theory \cite{arechigaUsingVerifiedControl2014,kabraCESARControlEnvelope2024}. 
Rather than synthesizing a single monolithic controller, in control-envelope synthesis, the goal is to compute a maximal set of control inputs whose execution traces satisfy a formal safety specification.
As a result, control envelopes enable decoupling safety from performance: the precomputed control envelope guarantees that \emph{every} admissible input meets the specification, while at runtime an optimization routine can freely sample concrete control signals with respect to a secondary performance objective.
A promising recent line of work \cite{Schaefer2024a} frames control envelope synthesis as the computation of the maximal \emph{robust control invariant}  sets (RCI) \cite{Bertsekas1972,Rungger2017,Rakovic2010_TAC}--- the set of states from which at each sampling instant, there exists a control action that keeps the system in it --- thereby guaranteeing safety by construction.
Building on this insight, the authors rely on efficient over-approximations of reachable sets as the main computational workhorse.
The efficiency of their computations is achieved by symbolic set representations, e.g., zonotopes \cite{Althoff2011f}, support functions \cite{Althoff2016c,Wetzlinger2023a}, or ellipsoids, that are either closed or can be tightly over-approximated under Minkowski addition and affine transformations, realizing an efficient computation of the reachable set \cite{Althoff2021b}.
However, despite their scalability, this reachability-based approach raises two concerns in safety-critical fields:
\begin{enumerate}
  \item Floating-point uncertainty: Current reachability tools perform every operation with finite precision, so the end result does not have any end-to-end guarantee that rounding errors have not led to a violation of the safety specification.
  \item Formally unverified implementations: The highly-tuned numerical kernels are optimized for speed, not transparency. Proving the correctness of their implementations is prohibitively expensive.
\end{enumerate}
Mission-critical control systems --- like aircraft, autonomous robots, or self-driving cars --- would benefit greatly from numerical pipelines that provide end-to-end correctness guarantees.
One way to obtain such rigorous guarantees is to use a theorem prover like KeYmaera~X \cite{Fulton2015} that implements differential dynamics logic (\dL) \cite{platzerCompleteUniformSubstitution2017b,platzerLogicalFoundationsCyberPhysical2018a}, a specialized logic designed for specification and deductive verification of cyber-physical systems.
Unlike numerical reachability tools \cite{Althoff2015a,chenFlowAnalyzerNonlinear2013,frehseSpaceExScalableVerification2011}, which aim to provide fast but conservative over-approximations of reachable sets, a theorem prover like KeYmaera X works with a symbolic proof calculus whose goal is to construct fully machine-checkable proofs that a control system satisfies its specification.
The key difficulty in combining techniques lies in the fundamental difference of formalisms of \dL and classical reachability analysis.
In reachability analysis, one works with highly specialized set representations that admit efficient "push-button" computations, but at the price of modelling restrictions.
In \dL, specifications are expressed as fully general semi-analytic formulas, resulting in a much richer specification language.
However, verification in this expressive framework typically demands substantial interactive proof effort for complicated systems.
As a result, tasks that are easy on one side, such as scalable numerical over-approximations, can be arduous on the other, such as manual proof construction, and vice versa.
Indeed, the theorem-proving and the reachability analysis research fields have been developed largely in parallel, with little cross-pollination.

Thus, our goal with this paper is to build a bridge between the two research fields.
Specifically, we formalize the control-envelope synthesis approach using RCIs in \dL.
By doing so, we enable end-to-end correctness guarantees for an envelope computed by an independent numerical reachability tool.
This integration of numerical reachability and deductive verification achieves something neither approach could accomplish alone: efficient numerical computation paired with fully rigorous, machine-checked correctness.
During formalization, we encountered the following challenges:
\begin{enumerate}
  \item \emph{Linking safety specifications to RCIs.} Although RCIs are commonly used as terminal constraints in model predictive control \cite{Blanchini1999,Schuermann2018}, no formal proof has yet shown that an RCI necessarily satisfies a given safety specification.
  \item \emph{Treatment of continuous dynamics.} The algorithms utilized by reachability tools are inherently delicate due to their numerical nature, therefore difficult to implement directly in \dL. Whilst the completeness of \dL \cite{platzerAxiomatizationCompactInitial} essentially guarantees that all true numerical properties can be deductively proven in principle, more efficient methods are desired for practical problems.
  \item \emph{Scalable verification of set containment.} Rigorous proofs of set containment --- a central step in reachability analysis methods --- involve large arithmetic formulas that are often intractable to verify without specialized methods. KeYmaera~X uses a general decision procedure for real-arithmetic formulas, which does not scale to large verification problems.
\end{enumerate}
\emph{Our Contribution} is to build a formal link between reachability analysis and deductive verification by introducing a \dL-based verification pipeline for control-envelopes addressing all three challenges.
Specifically, we provide a syntactic derivation showing that if a control envelope is a robust control invariant set, then it automatically satisfies the \dL specification. 
In addition, we leverage Taylor Models in \dL, which is general enough to validate the invariant sets synthesized by numerical methods while also having rigorous error bounds that are deductively proven in \dL.
Finally, to overcome the scalability bottleneck in set containment proofs, we focus on zonotopic control envelopes and extend a known witness theorem for zonotope containment. 
This enables fast floating-point search for a candidate witness followed by a lightweight certification in KeYmaera X without relying on the slow, general-purpose decision procedure.
\section{Preliminaries} \label{sec:preliminaries}
This section first recalls the control envelope synthesis problem and then distills the basic principles of differential dynamic logic that support our verification approach.
Let \(\setSafety \subseteq \R^n\) denote the set of admissible states and \(\setControl \subset \R^m\) be the set of admissible control inputs.
For a \emph{sampling period} \(\constSamplingTime > 0\) the \emph{sampled-data system} is given by
\begin{equation} \label{eq:sampled_data}
  x'(t)=f(x(t),u_{\lfloor t/\constSamplingTime\rfloor}), \quad x_0 \in \setInitial,
\end{equation}
where the control input is zero-order-hold: for each integer \(k \geq 0\) we sample at \(t_k = k \constSamplingTime\) a feedback law \(\mu : \R^n \rightarrow \R^m\) generates \(u_k \define \mu(x(t_k))\), which is held constant on \([t_k, t_{k+1}]\).
In the classical control, problem we seek to find one concrete feedback law \(\mu : \R^n \rightarrow \R^m\) such that every control input is admissible,
\begin{equation} \label{eq:control}
  \mu(x(k\constSamplingTime)) \in \setControl,
\end{equation}
and every trajectory of the closed-loop system remains in the admissible set of states,
\begin{equation} \label{eq:safety}
  x(t) \in \setSafety
\end{equation}
for all \(t \geq 0\).
The control envelope problem lifts this problem from finding a single control law to a family of laws.
Concretely, we seek to construct a relation \(\setEnvelope \subseteq \R^n \times \R^m\) between states and control inputs such that \emph{every} feedback law whose sampled control inputs stay within envelope, automatically satisfy the admissibility and safety specifications.

In reachability analysis, one over-approximates ODE trajectories by a set-valued abstraction, a perspective that yields efficient computational properties.
This viewpoint naturally leads to the following notion of reachable sets.
\begin{definition}[Reachable set]
	Let \(\constSamplingTime > 0\) be the sampling period and let $\setEnvelope \;\subseteq\; \R^n \times \R^m$ be a control envelope: \(\setEnvelope_x \mathrel{:=} \{u \in \R^m \mid (x,u) \in \setEnvelope\}\) denotes the set of control outputs at \(x \in \R^n\). Given an initial set $\setInitial$, the \emph{reachable set} at time $t \in \cinterval[0][\constSamplingTime]$ is the set of states
	\begin{align*}
			\setReach \mathrel{:=} 
      \left\{x \in \R^n \mid \exists x_0\in \setInitial\exists u \in \setEnvelope_{x_0}  : x = x_0 + \int_0^{t} f(x(s),\uctrl)ds\right\}.
		\label{eq:exactReachSetTP}
	\end{align*} 
	The reachable set over the time interval $\cinterval[0][{t}]$ is the union of reachable sets, i.e.,
	\begin{equation*}
		\setReach[\cinterval[0][t]][\setInitial] = \bigcup_{s \in \cinterval[0][{t}]} \setReach[s][\setInitial].
		\label{eq:exactReachSetTi}
	\end{equation*}
\end{definition}
The difficulty with reachable-set computations is that we can only evaluate them over finite time horizons.
How can we verify that a proposed control envelope \(\setEnvelope\) satisfies the safety property \eqref{eq:safety} for all \(t \geq 0\)?
The key is to use an inductive argument: leverage finite-horizon reachable-set computations to establish an invariant that guarantees safety over an infinite time horizon.
This leads us to the definition of robust control invariants.
\begin{definition}[Robust control invariant set~\cite{Schaefer2024a}] \label{def:RCI}
  A set \(\setRCI \subset \R^n\) is called a \emph{robust control invariant set} if there exists a control envelope \(\setEnvelope \;\subset\; \R^n \times \R^m\) such that
  \begin{enumerate}
    \item \emph{One-step invariance}: \(\setReach[\constSamplingTime][\setRCI] \;\subseteq\; \setRCI,\) \label{item:invariance}
    \item \emph{One-step safety}: \(\setReach[\cinterval[0][\constSamplingTime]][\setRCI] \;\subseteq\; \setSafety ,\) \label{item:safety}
    \item \emph{Control-admissibility}: \(\forall x_0 \in \setRCI : \setEnvelope_{x_0} \;\subseteq\; \setControl. \) \label{item:control}
  \end{enumerate}
\end{definition}
The central claim is that the existence of a robust control invariant \(\setRCI\) and its associated control envelope \(\setEnvelope\), implies the safety property \eqref{eq:safety} for all \(t \geq 0\).
We prove this claim formally in \cref{sec:verification}.
\subsection{Differential Dynamic Logic}
To formally verify a control envelope, we must first introduce differential dynamic logic: its hybrid-program modeling language, its formula specification language, and the proof calculus that enables deductive verification.
For a more detailed introduction the reader is referred to the literature \cite{platzerCompleteUniformSubstitution2017b,platzerLogicalFoundationsCyberPhysical2018a}.
\paragraph{Hybrid Programs.} The language of hybrid programs is generated by the following grammar, where \(x\) is a variable, \(\termA\) is a \dL term, \(\formulaB\) is a formula of first-order real arithmetic:
\[
  \programA, \programB \mathrel{::=} x:=\termA \vl x:= * \vl \progOde \vl ?\formulaA \vl \programA; \programB \vl \programA \cup \programB \vl \programA^*.
\]
Continuous dynamics are modeled by \(x' = f(x) \;\&\; Q\) evolving in the domain \(Q\).
Discrete dynamics are modeled by \emph{assignments} \(x := e\),  which instantaneously assign the term \(e\) to \(x\) and \emph{tests} \(?Q\), which check whether the formula \(Q\) is true in the current state.
The \emph{nondeterministic assignment} \(x := *\) assigns an arbitrary value to \(x\).
Throughout this paper we always use such assignments in conjunction with a test \(?Z(x)\), to model non-deterministic assignment restricted to a formula \(Z(x)\).
To combine the discrete and continuous fragments, there are three program combinators: \emph{sequential composition} \(\alpha ; \beta\), which first runs \(\alpha\) then \(\beta\); \emph{nondeterministic choice} \(\alpha \cup \beta\), which runs either \(\alpha\) or \(\beta\); and finally \emph{nondeterministic repetitions} \(\alpha^*\), which repeats \(\alpha\) an arbitrary number of times.
\paragraph{Formulas.}
The formulas of \dL are defined by the following grammar where
  \(\termA, \termB\) are terms, \(\formulaA, \formulaB\) are formulas, \(x\) is a variable, and \(\programA\) is a hybrid program:
\[
  \formulaA, \formulaB ::= \termA \leq \termB \vl \neg \formulaA \vl \formulaA \land \formulaB
  \vl \formulaA \rightarrow \formulaB \vl \forall x  \formulaA
  \vl [\alpha]\formulaA.
\]
A \dL formula combines first-order arithmetic with model operators that refers to program behavior.
Atomic formulas are inequalities \(\termA \leq \termB\) between real-valued terms.
These atomic formulas are composed with Boolean connectives such as \(\neg\), \(\land\) and \(\rightarrow\), together with the  first-order quantifier \(\forall\).
Connectives such as \(\lor\) and the quantifier \(\exists\) are definable from these primitives.
The only distinctive construct is the \emph{box modality} \([\alpha]\formulaA\), which asserts that after every execution of the hybrid program \(\alpha\) the post-condition \(\formulaA\) holds.
\subsection{Deductive Verification}
Formula verification in \dL is carried out within a sequent-calculus framework built on sound axioms and inference rules.
We write \(\Gamma \vdash \Delta\) if the formula \(\Delta\) is provable from the assumption \(\Gamma\) in the \dL proof calculus.
The calculus includes propositional rules such as
\ruleset{impliesRight}{\(\rightarrow\)\textup{R}}
\ruleset{andRight}{\(\land\)\textup{R}}
\begin{mathpar}
  \ruletag{impliesRight}
  \AX$\fCenter \Gamma, P \vdash Q, \Delta$
  \UI$\fCenter \Gamma \vdash P \rightarrow Q, \Delta$
  \DisplayProof
  \and
  \ruletag{andRight}
  \AX$\fCenter \Gamma \vdash P, \Delta$
  \AX$\fCenter \Gamma \vdash Q, \Delta$
  \BI$\fCenter \Gamma \vdash P \land Q, \Delta$
  \DisplayProof
\end{mathpar}
The \ruleref{impliesRight} rule decomposes the implication \(P \rightarrow Q\) by adding \(P\) in the list of assumptions, while the \ruleref{andRight} rule splits the conjunctive formula \(P \land Q\) into two separate proof goals.
Similarly, there are rules that decompose every other logical construct.
These inference rules are applied bottom-up, but provability is read top-down: if the premises at the top of each rule are provable, then the conclusion is provable as well. 

Box modalities of hybrid program are treated with an axiomatic proof calculus that strips away the program structure step by step, leaving simpler proof goals.
A few representative axioms illustrate this:
\axset{boxSeq}{[\,;\,]}
\axset{boxChoice}{$[\cup]$}
\begin{flalign*}
  \axtag{boxSeq} & \blue{[\alpha ; \beta]\formulaA} \leftrightarrow [\alpha][\beta]\formulaA  \\
  \axtag{boxChoice} & \blue{[\alpha \cup \beta]\formulaA} \leftrightarrow [\alpha]\formulaA \land [\beta]\formulaA 
\end{flalign*}
The \axref{boxSeq} axiom unfolds the sequential composition into successive modalities, while the \axref{boxChoice} axiom reduces the nondeterministic choice \(\alpha \cup \beta\) to a conjunction over two branches.

\begin{wrapfigure}{r}{0.52\textwidth}
  \centering
  \begin{minipage}[c]{0.04\textwidth}
    $\overset{\textbf{Deduction}}{\Vast\uparrow}$
  \end{minipage}%
  \hfill
  \begin{minipage}[c]{0.48\textwidth}
    \begin{prooftree}
      \AX$\fCenter \vdash Q_1(x)$
      \AX$\dots \fCenter$
      \AX$\vdash \;\fCenter Q_n(x)$
      \TI$\vdots \;\, \fCenter$
      \UI$\Gamma \vdash \;\fCenter [\alpha][\beta]\varphi$
      \LeftLabel{\axref{boxSeq}}
      \UI$\Gamma \vdash \;\fCenter [\alpha;\beta]\varphi$
      \LeftLabel{\ruleref{impliesRight}}
      \UI$\vdash \Gamma \rightarrow \;\fCenter [\alpha;\beta]\varphi$
    \end{prooftree}
  \end{minipage}
  \vspace*{-1.5\baselineskip}%
\end{wrapfigure}
By combining the inference rules with the axioms, we can derive new theorems, and every step of the resulting derivation can be mechanically verified by a proof checker.
Transformation steps on logical connectives and box modalities are applied until the renaming goal is purely arithmetic.
Then, we can invoke a trusted decision procedure, such as quantifier elimination.
In the example, we first apply the propositional rule \ruleref{impliesRight}, then the \axref{boxSeq} axiom.
After a few further transformation steps we reach purely arithmetic formulas \(Q_1(x), \dots, Q_n(x)\), at which point we hand the proof obligations to the trusted decision procedure, finishing the proof.

\section{Control Envelope Verification Framework} \label{sec:verification}
In this section, we first formalize the control \emph{envelope} synthesis problem from \cref{sec:preliminaries} in \dL.
To do so, we first model the sampled-data systems \eqref{eq:sampled_data} as a hybrid program.
We can then state the safety property \eqref{eq:safety} formally in \dL.
Once these components are formalized, we present several key theorems which, taken together, yield a deductive proof of the safety property.
This allows us to take a control envelope produced numerically and rigorously verify its adherence to the specification.
\subsection{Control Envelope Synthesis Problem in \dL}
In the following we adopt the conventions: for a function symbol \(h\), let \(\timeDerivative{h}\) denote its time derivative, and we use \(\nabla_x\) to denote the syntactic partial gradient with respect to \(x\).
We use \(\normInf\) to denote the infinity norm of \(\R^n\).
Hereafter we write the corresponding uppercase predicates, e.g. \(\formulaEnvelope(x,u)\), \(\formulaRCI\), \(\formulaInitial\), \(\formulaSafety\) and \(U(u)\), to denote their syntactic formula counterparts in \dL.
In \cref{sec:preliminaries}, we used a calligraphic font to denote semantic sets, e.g. \(\setEnvelope\) for the control envelope, \(\setRCI\) for the robust control invariant set, \(\setInitial\) for the set of initial states, \(\setSafety\) for the safety set and \(\setControl\) for the control constraint set.

We now formulate the sampled-data systems.
Let \(x\) be the state vector, \(u\) the control input and \(\formulaEnvelope(x,u)\) a control envelope formula, and \(\constSamplingTime > 0\) the sampling period.
We define the \emph{initialization}, \emph{controller} and
\emph{plant} components as follows:
\begin{align*}
  \init &\equiv u:=*; ?\exists x \formulaEnvelope(x, u), \\
  \ctrl &\equiv (?t = \constSamplingTime; u:=*; ?\formulaEnvelope(x, u); t:= 0 \;) \cup (? t
  \neq \constSamplingTime), \\
  \plant &\equiv \{\progOdeTimeControl\}.
\end{align*}
The controller nondeterministically chooses between two branches. If the current time tt has not yet reached the sampling period $\constSamplingTime$, it does nothing.
Once the sampling instant is reached, a new control action is chosen nondeterministically from within the control envelope \(\formulaEnvelope\) based on the current state \(x\).
Note that this controller is an abstraction of any concrete implementation.
By isolating only the aspects relevant to safety verification, we simplify reasoning in \dL.
Any actual controller would, of course, adhere to additional performance criteria.
The plant is modeled by the differential equation \(x'=f(x,u),t'=1\) subject to the constraint \(t \leq \constSamplingTime\) to ensure a duration of evolution of at most \(\constSamplingTime\).
In \dL, differential equations are non-deterministic: the continuous evolution may stop at any state that still satisfies the constraint \(t \leq \constSamplingTime\).
The \dL \emph{sampled-data system} for one sampling period is
\begin{align*}
  \sys &\equiv \ctrl;\plant.
\end{align*}
The corresponding \dL closed-loop system is the nondeterministic repetition of this program \(\sys^*\).
Thus, the \dL \emph{control envelope synthesis} problem is to determine an envelope predicate \(\formulaEnvelope(x,u)\) such that the \emph{control-admissibility formula}
\begin{align} \label{fml:control_admissibility}
  \exists x \formulaEnvelope(x,u) \rightarrow U(u),
\end{align}
and the \emph{closed-loop safety formula}
\begin{align} \label{fml:system_safety}
  X_0(x) \land t = 0 \rightarrow [\init;\sys^*]X(x)
\end{align}
are valid, i.e., true in all states.
These two formulas are the formalized counterparts to \eqref{eq:control} and \eqref{eq:safety}.
\subsection{Deductive Verification of Control Envelopes}
Having posed the control envelope synthesis problem in \dL, we can take an envelope computed by an unverified numerical tool, formalize its specification in \dL
and certify its correctness through an independent verification process.
When carrying out deductive verification, several key considerations arise:
The control-admissibility formula \eqref{fml:control_admissibility} is purely a first-order statement over the reals, so it can be proved without invoking any of \dL's hybrid system axioms.
In contrast, the system-safety formula \eqref{fml:system_safety} is more challenging: it combines discrete and continuous dynamics, a nondeterministic repetition, and may contain a high-dimensional control envelope.
To tackle this challenge for generic sampled-data system, we present a systematic verification strategy.
Specifically, we present a sequence of theorems that, when composed, yield a rigorous proof of the desired safety property.
The following five theorems outline the high-level structure of our argument:
\begin{enumerate}
  \item \cref{thm:rci}: Relates closed-loop safety \eqref{fml:system_safety} to the invariance and safety properties of robust control invariant sets.
  \item \cref{thm:taylor_model}: Connects finite-horizon box modalities for continuous systems to Taylor models.
  \item \cref{thm:reach_discrete}: A special case of \cref{thm:taylor_model} that establishes the RCI invariance condition for zonotopic control envelopes.
  \item \cref{thm:reach_interval}: A special case of \cref{thm:taylor_model} used to establish the RCI safety condition for zonotopic control envelopes.
  \item \cref{thm:zonotope}: Connects the zonotope-containment problem to an efficiently implementable witness-checking problem.
\end{enumerate}

The first theorem shows that, in \dL, the invariance and safety property of robust control invariants imply the closed-loop safety property \eqref{fml:system_safety}:
\begin{theoremrep} \label{thm:rci}
  The following is a derived proof rule of \dL:
  \def\ScoreOverhangRight{2cm}
  \label{thm:loop-safety}
  \begin{prooftree}
      \AX$ \formulaEnvelope(x, u), t = 0 \fCenter\;\vdash  [\plant]X(x)$
      \noLine
      \UI$ \formulaEnvelope(x, u), t = 0 \fCenter\;\vdash  [\plant](t =\Delta t\rightarrow \exists u \formulaEnvelope(x, u))$
      \UI$ X_0(x), t = 0 \fCenter\;\vdash [\init;\sys^*]X(x)$
    \end{prooftree}
\end{theoremrep}
\begin{appendixproof}
  We prove each rule:
  \begin{enumerate}
    \item[\link{rl:rci}] We have
      \begin{prooftree}
      \AXC{\(X_0(x), t = 0, S(x,u) \vdash [\sys^*]X(x)\)}
      \LeftLabel{\ruleref{impliesRight},\ruleref{forallRight}}
      \UIC{\(X_0(x), t = 0 \vdash \forall u (S(x,u) \rightarrow
      [\sys^*]X(x))\)}
      \LeftLabel{\axref{boxTest},\axref{boxNonDetAssign}}
      \UIC{\(\Gamma, X_0(x), t = 0 \vdash [\init;\sys^*]X(x), \Delta\)}
    \end{prooftree}
    The rest follows from Lemma \ref{lemma:loop}.
\end{enumerate}
\end{appendixproof}
Informally, the bottom premise states that whenever an initial control-state pair satisfies \(\formulaEnvelope(x,u)\), then, after one continuous evolution lasting \(\constSamplingTime\), we can choose a new control action to remain in \(\formulaEnvelope(x,u)\).
The top premise encodes the safety property: if the start in \(\formulaEnvelope(x,u)\), then we always remain in the safety constraint \(X(x)\) after running the plant.
The safety property is enforced continuously; not just at the sampling points.
Taken together, these two properties entail the safety of the closed-loop system for an arbitrary number of execution steps.
The premises are much easier to discharge, because they only involve ordinary differential equations.
The two main challenges in proving the premises in \cref{thm:rci} are 
\begin{enumerate}
  \item constructing rigorous reachable set over-approximations,
  \item efficiently discharging the resulting proof obligations, which involve large arithmetic formulas.
\end{enumerate}
We address the first challenge by introducing generalized Taylor models and showing how they can be derived in \dL.
Before introducing Taylor models, we briefly review the necessary background Picard iteration and interval arithmetic.
\paragraph{Picard iteration}
Picard iteration is a classical technique that successively constructs  polynomial approximations to the solutions of ordinary differential equations.
They can be carried out to arbitrarily high order, resulting approximations of whatever precision is required.
Let \(h\) be a function symbol and let \(\lambda\) be a variable vector such that \(\normInf[\lambda] \leq 1\).
Further, let \(x' = f(x)\) be an ODE system with parameter-dependet initial value \(x_0 = h(\lambda)\).
The sequence of \emph{Picard iterates} \((p_k)_{k\geq 0}\) is recursively defined by
\begin{align*}
  p_0(t,\lambda) \mathrel{:=} h(\lambda), \quad p_{k+1}(t,\lambda) \mathrel{:=} h(\lambda) + \int_{0}^{t} f(p_k(s,\lambda))ds, \quad k \geq 0.
\end{align*}
\paragraph{Interval arithmetic.}
A foundational concept we are making use of to discharge arithmetic proof obligations is \emph{interval arithmetic} \cite{mooreIntroductionIntervalAnalysis2009}.
We write \(\setIntervals^n\) to denote the set of rational intervals in the \(n\)-dimensional vector space \(\Q^n\).
We write \(\intA = [\opIntLower{\intA},\opIntUpper{\intA}] \in \setIntervals^n\), where \(\opIntLower{\intA} \in \Q^n \) and \(\opIntUpper{\intA} \in \Q^n\) are the lower and upper bound respectively.
In order to simplify, notation we write
\begin{align*}
  \opIntMid{\intA} &\mathrel{:=} \frac{\opIntUpper{\intA} + \opIntLower{\intA}}{2},\\
  \opIntRad{\intA} &\mathrel{:=} \opIntUpper{\intA} - \opIntLower{\intA},
\end{align*}
for an intervals \emph{mid-point} and \emph{radius}.
For a variable \(x\) and concrete \dL term \(e\), we write \(\opIntEval{e}{\intA}{x}\)
to denote its interval evaluation with respect to the interval \(\intA\).
\paragraph{Taylor models.} Equipped with interval arithmetic and Picard polynomial approximations, we are now ready to define Taylor models in \dL\cite{berzVerifiedIntegrationODEs1998}.
\begin{definition}[Taylor models]
  Let \(p(t,\lambda)\) be a concrete \dL term and \(\intA(t)\) a \dL interval. Then a \emph{\dL Taylor model} is a tuple \((p, \intA)\) whose associated formula is
  \begin{align*}
    \TM{p}{\intA}{x,\lambda} \;\equiv\; \opIntLower{\intA}(t) \leq x - p(t,\lambda) \leq \opIntUpper{\intA}(t).
  \end{align*}
  For its derivative formula, we write 
  \begin{align*}
    &\TMDerivative{p}{\intA}{x, \lambda} \; \equiv \; \timeDerivative{\opIntLower{\intA}}(t) < f(x) - \timeDerivative{p}(t, \lambda) < \timeDerivative{\opIntUpper{\intA}}(t).
  \end{align*}
\end{definition}
\begin{theoremrep} \label{thm:taylor_model}
    Let \((p,\intA)\) be a Taylor model.
    Let \(h\) be a function symbol and \(X_0(x) \equiv \exists \lambda \left( x = h(\lambda) \land \normInf[\lambda] \leq 1\right)\).
    The following is a sound derived proof rule of \dL:
    \begin{prooftree}
      \AX$\fCenter \exists \lambda \exists t (\TM{p}{\intA}{x,\lambda} \land 0 \leq t \leq \Delta t \land \normInf[\lambda] \leq 1) \;\vdash \formulaA(x)$
      \noLine
      \UI$\fCenter X_0(x) \vdash \exists \lambda \left(\TM{p}{\intA}{x,\lambda}[0] \land \normInf[\lambda] \leq 1\right)$
      \noLine
      \UI$\fCenter \TM{p}{\intA}{x,\lambda},\; 0 \leq t \leq \Delta t,\; \normInf[\lambda] \leq 1\vdash \TMDerivative{p}{\intA}{x,\lambda}$
      \UI$\fCenter \fCenter X_0(x),\; t = 0 \vdash \; [\progOdeTime]\formulaA(x)$
    \end{prooftree}
\end{theoremrep}
\begin{appendixproof}[Sketch.]
  In this proof we use the following notation to denote the closed version of the derivative formula:
  \[
    \TMDerivative{p}{\intA}{x, \lambda}^{\leq} \; \equiv \; \timeDerivative{\opIntLower{\intA}}(t) \leq f(x) - \timeDerivative{p}(t, \lambda) \leq \timeDerivative{\opIntUpper{\intA}}(t).
  \]
  Note that the theorem follows by establish the derivability of the following proof rule:
  \begin{prooftree}
      \AX$\fCenter X_0(x) \vdash \exists \lambda \left(\TM{p}{\intA}{x,\lambda}[0] \land \normInf[\lambda] \leq 1\right)$
      \noLine
      \UI$\fCenter \TM{p}{\intA}{x, \lambda},\; 0 \leq t \leq \Delta t \;\vdash \TMDerivative{p}{\intA}{x, \lambda}$
      \UI$\fCenter \fCenter X_0(x),\; t = 0 \vdash  [\progOdeTime]\TM{p}{\intA}{x, \lambda}$
    \end{prooftree}
  as the desired conclusion then follows by a sound application of differential weakening. To derive this proof rule, note $\TM{p}{\intA}{x, \lambda}$ only contains non-strict inequalities, so we may soundly apply closed differential induction.
  \begin{prooftree}
      \AX$\fCenter X_0(x) \vdash \exists \lambda \left(\TM{p}{\intA}{x,\lambda}[0] \land \normInf[\lambda] \leq 1\right)$
      \noLine
      \UI$\fCenter  \;\vdash [\progOdeTime \land \TM{p}{\intA}{x, \lambda}]\TMDerivative{p}{I}{x, \lambda}$
      \UI$\fCenter \fCenter X_0(x),\; t = 0 \vdash  [\progOdeTime]\TM{p}{\intA}{x, \lambda}$
    \end{prooftree}
    Note that the top premise is part of our assumption, so it remains to handle the second premise. Applying differential weakening yields 
    \begin{prooftree}
        \AX$\fCenter \TM{p}{\intA}{x, \lambda},\; 0 \leq t \leq \Delta t ,\; \normInf[\lambda] \leq 1 \vdash \TMDerivative{p}{\intA}{x, \lambda}$
      \UI$\fCenter  \;\vdash [\progOdeTime \land \TM{p}{\intA}{x, \lambda}]\TMDerivative{p}{I}{x, \lambda}$
    \end{prooftree}
    thereby reducing the conclusion to one of our assumptions, hence the derivation is complete.
    \qed
\end{appendixproof}
\paragraph{Zonotope reachable set} We now focus our attention on the zonotope case and set out to prove premise 1 of \cref{thm:rci}.
Our strategy is to successively rewrite the proof goal until it is an instance of zonotope containment. Once in that form, a lightweight numerical solver can supply a witness that one zonotope is contained in the other.
\begin{definition}[Zonotope]\label{def:zonotope}
  Let \(\matrixZonoA\) be an \(n \times p\) generator \dL matrix,  
  \(\vecZonoA\) an \(n\) center \dL vector, and let  
  \(\vecZonoCoeff\) be a \(p\) \dL vector.
  The \emph{zonotope formula} associated with \(\formulaZonoA\) is given by
  \[
    \formulaZonoA(x) \;\equiv\; \exists\vecZonoCoeff
      \bigl( x = \vecZonoA + \matrixZonoA\,\vecZonoCoeff \land \normInf[\vecZonoCoeff] \le 1 \bigr).
  \]
\end{definition}
In order to simplify arithmetic reasoning, our goal is to linearize the Picard iterates.
This helps us to over-approximate the reachable set at the time instance \(\constSamplingTime\) by a zonotope.
We have the following theorem:
\begin{lemmarep}[Linear interval abstraction]\label{lemma:linearAbstraction}
  Let \(p(x)\) be a concrete \dL polynomial in \(x\).
  Then, for the interval \(\intA \in \setIntervals^n\) and the interval remainder
  \[
  \intRemainder = \opIntEval{p(x) - p(0) - \nabla_{x}p(0)^\top x}{\intA}{x}
  \]
  the following formula is a sound axiom of \dL:
  \begin{align*}
    \opIntLower{\intA} \leq x \leq \opIntUpper{\intA} \rightarrow \exists \xi (p(x) = p(0) + \nabla_{x}p(0)^\top x + \opIntMid{\intRemainder} + \frac{1}{2}\opIntRad{\intRemainder} \xi \land \normInf[\xi] \leq 1 ).
  \end{align*}
\end{lemmarep}
\begin{proof}
  Let \(x \in \R^n\). Note that for an interval \(\intA \in \setIntervals^n\) we have \(x \in \intA\) iff \[\exists \eta (x = \opIntMid{\intA} + \frac{1}{2}\opIntRad{\intA}^\top\eta \land \normInf[\eta] \leq 1).\]
  By assumption, we have \(x \in \intA\). This implies
  \begin{align*}
    p(x) &= p(0) + \nabla_{x}p(0)^\top x + (p(x) - p(0) - \nabla_{x}p(0)^\top x) \\
    &\in p(0) + \nabla_{x}p(0)^\top x + \opIntEval{p(x) - p(0) - \nabla_{x}p(0)^\top x}{\intA}{x}.
  \end{align*}
  Finally, using the formula representation from above, for some \(\xi \in \R^n\) with \(\normInf[\xi] \leq 1\) we have
  \begin{align*}
    p(x) = p(0) + \nabla_{x}p(0)^\top x + \opIntMid{\intRemainder} + \frac{1}{2}\opIntRad{\intRemainder}^\top\xi.
  \end{align*}
  \qed
\end{proof}
\begin{theoremrep}[Zonotope reachable set for discrete time instance] \label{thm:reach_discrete}
  Let \((p, \intA)\) be a Taylor model of the ODE system \(x'=f(x) \;\&\; t \leq \constSamplingTime\).
  Let \[
  \intRemainder = \opIntEval{p(t,\lambda) - p(t,0) - \nabla_{\lambda}p(t,0)^\top \lambda}{[0,\constSamplingTime] \times [-1,1]^n}{(t,\lambda)}
  \] be the remainder of its linear interval abstraction.
  Finally, we define the \dL center vector and generator matrix
  \begin{align*}
    \vecZonoB &\define p(\constSamplingTime, 0)  + \opIntMid{\intA(\constSamplingTime)} + \opIntMid{\intRemainder}, \\
    \matrixZonoB &\define [\nabla_{\lambda}p(\constSamplingTime, 0),\;\frac{1}{2} \opIntRad{\intA(\constSamplingTime)},\;\frac{1}{2}\opIntRad{\intRemainder}].
  \end{align*}
  
  Then, the following is a sound derived proof rule of dL:
  \begin{prooftree}
    \AX$\fCenter \formulaZonoB(x,u) \vdash \zonotope{\vecZonoA_x}{\matrixZonoA_x}(x,u)$
    \noLine
    \UI$\fCenter \formulaZonoA(x,u) \vdash \exists \lambda \left(\TM{p}{\intA}{x,u,\lambda}[0] \land \normInf[\lambda] \leq 1\right)$
    \noLine
      \UI$\fCenter \TM{p}{\intA}{x, u,\lambda},\; 0 \leq t \leq \Delta t,\; \normInf[\lambda] \leq 1\vdash \TMDerivative{p}{\intA}{x, u,\lambda}$
    \UI$\fCenter \formulaZonoA(x,u), t=0 \vdash  [\plant]\left(t = \Delta t
        \rightarrow \exists u \, \formulaZonoA(x,u)\right)$
  \end{prooftree}
\end{theoremrep}
\begin{proof}
  First, by equivalence rewriting we have 
  \(\exists u \formulaZonoA(x,u) \equiv \exists u \zonotope{(\vecZonoA_x, \vecZonoA_u)}{(\matrixZonoA_x, \matrixZonoA_u)} \equiv \zonotope{\vecZonoA_x}{\matrixZonoA_x}(x).\)
  Then, we rewrite the Taylor model formula:
  \begin{align*}
    \TM{p}{I}{x,u}[\constSamplingTime] &\equiv \exists \lambda \left(\opIntLower{\intA}(\constSamplingTime) \leq (x,u) - p(\constSamplingTime,\lambda) \leq \opIntUpper{\intA}(\constSamplingTime) \land \normInf[\lambda] \leq 1\right) \\
    &\equiv \exists \xi \exists \lambda \left((x,u) = p(\constSamplingTime,\lambda) + \opIntMid{\intA(\constSamplingTime)} + \frac{1}{2} \opIntRad{\intA(\constSamplingTime)} \xi \land \normInf[(\lambda,\xi)] \leq 1\right).
  \end{align*}
  Now, we apply \cref{lemma:linearAbstraction} to the polynomial constraining \((x,u)\) to obtain:
  \begin{prooftree}
    \AX$\zonotope{p(\constSamplingTime, 0) + \opIntMid{\intRemainder} + \opIntMid{\intA(\constSamplingTime)}}{[\nabla_{\lambda}p(\constSamplingTime, 0), \frac{1}{2}\opIntMid{\intRemainder},\frac{1}{2} \opIntMid{\intA(\constSamplingTime)}]}(x) \fCenter \vdash  \zonotope{\vecZonoA_x}{\matrixZonoA_x}(x)$
    \UI$\TM{p}{I}{x,u}[\constSamplingTime] \fCenter \vdash \exists u \formulaZonoA(x,u)$
  \end{prooftree}
  \qed
\end{proof}
We now turn our attention to second premise in \cref{thm:rci}.
Unlike the first premise, which involves a single sampling instant, this condition must hold for every \(t \in [0,\constSamplingTime]\).
In other words, we must enclose the entire time-tube of states captured by \(\setReach[\cinterval[0][\constSamplingTime]]\).
We tackle this problem by constructing a single zonotope that encloses the reachable set at every time instant \(t \in [0,\constSamplingTime]\).
\begin{theoremrep}[Zonotope reachable set for time intervals] \label{thm:reach_interval}
  Let \((p,\intA)\) be a Taylor model for the ODE system \(x'=f(x) \; \& \; t \leq \constSamplingTime\). We define the remainder polynomial
  \begin{align*}
    &r(t,\lambda,\xi)
      \define p(t,\lambda) + \opIntMid{\intA(t)} + \frac12\,\opIntRad{\intA(t)}^{\top}\xi - p(0,0) - \opIntMid{\intA(0)} \\
             &\quad - \timeDerivative{p}(0,0)\,t
             - \timeDerivative{\opIntMid{\intA(t)}}\bigl|_{t=0}\,t
             - \nabla_{\lambda}p(0,0)^{\top}\lambda
             - \tfrac12\,\opIntRad{\intA(0)}^{\top}\xi .
  \end{align*}
  To bound its range over the domain \(\intB \define [0,\constSamplingTime]\times[-1,1]^n\times[-1,1]^n\),
  we introduce the error interval
  \[
    \intRemainder
      \define \opIntEval{\,r(t,\lambda,\xi)\,}
               {\intB}
               {(t,\lambda,\xi)}.
  \]
  Finally, we define the \dL center vector and generator matrix
  \begin{align*}
    \vecZonoB &\define p(0,0) + \opIntMid{\intA(0)} + \opIntMid{\intRemainder}, \\
    \matrixZonoB &\define [\timeDerivative{p}(0,0) + \timeDerivative{\opIntMid{\intA(t)}}|_{t=0}, \; \nabla_\lambda p(0,0), \; \frac{1}{2} \opIntRad{\intA(0)}, \frac{1}{2}\opIntRad{\intRemainder}].
  \end{align*}
  Then, the following is a sound derived proof rule of \dL:
  \begin{prooftree}
    \AX$\fCenter \formulaZonoB(x) \vdash X(x)$
    \noLine
    \UI$\fCenter \formulaZonoA(x,u) \vdash \exists \lambda \left(\TM{p}{\intA}{x,u,\lambda}[0] \land \normInf[\lambda] \leq 1\right)$
    \noLine
    \UI$\fCenter \TM{p}{\intA}{x, u,\lambda},\; 0 \leq t \leq \Delta t,\; \normInf[\lambda] \leq 1\vdash \TMDerivative{p}{\intA}{x, u,\lambda}$
    \UI$\fCenter \formulaZonoA(x,u), t=0 \vdash  [\plant]X(x)$
  \end{prooftree}
\end{theoremrep}
\begin{proof}
  First, by equivalent rewriting we obtain
  \begin{align*}
    \TM{p}{I}{x,u} &\equiv \exists \lambda \left(\opIntLower{\intA}(t) \leq (x,u) - p(t,\lambda) \leq \opIntUpper{\intA}(t) \land \normInf[\lambda] \leq 1\right) \\
    &\equiv \exists \xi \exists \lambda \left((x,u) = p(t,\lambda) + \opIntMid{\intA(t)} + \frac{1}{2} \opIntRad{\intA(t)} \xi \land \normInf[(\lambda,\xi)] \leq 1\right).
  \end{align*}
  Then, we can use \cref{lemma:linearAbstraction} on the polynomial constraining \((x,u)\):
  \begin{align*}
    q(t,\lambda,\xi) &= p(t,\lambda) + \opIntMid{\intA(t)} + \frac{1}{2} \opIntRad{\intA(t)}^\top \xi \\
    &\in q(0,0,0) + \nabla_{(t,\lambda, \xi)}q(0,0,0) \cdot (t,\lambda,\xi)+ \opIntMid{\intRemainder} + \frac{1}{2}\opIntRad{\intRemainder}\eta.
  \end{align*}
  Finally, using this over-approximation, by a cut and some contextual reasoning we obtain
  \begin{prooftree}
    \AX$\zonotope{p(t, 0) + \opIntMid{\intRemainder} + \opIntMid{\intA(t)}}{(\nabla_{\lambda}p(t, 0), \frac{1}{2}\opIntMid{\intRemainder},\frac{1}{2} \opIntMid{\intA(t)})}(x) \fCenter \vdash X(x)$
    \UI$\TM{p}{I}{x,u} \fCenter \vdash X(x)$
  \end{prooftree}
  \qed
\end{proof}
  
\paragraph{Witness checks for arithmetic goals.}
Recall that our earlier theorems translate the reachability problem into a single question of zonotope containment.
In practice, the zonotopes that bound control envelopes usually involve on the order of ten to twenty generators.
Consequently, the resulting formulas can become quite large.
A straightforward quantifier-elimination procedure is not a good fit for this problem, because its computational costs are prohibitive.
A common approach in reachability analysis is to invoke a witness theorem \cite[Cor. 4]{sadraddiniLinearEncodingsPolytope2019}: for zonotopes \(\formulaZonoA(x), \; \formulaZonoB(x)\) the containment \[\forall x \big(\formulaZonoA(x) \rightarrow \formulaZonoB(x)\big)\] holds whenever one can find a matrix \(\Gamma\) and vector \(\beta\) that satisfy 
\begin{equation} \label{eq:sadraddini_zono_cont}
	\matrixZonoB\Gamma = \matrixZonoA,\quad \vecZonoB - \vecZonoA = \matrixZonoB\beta, \quad \normInf[(\Gamma, \beta)] \leq 1.
\end{equation}

These linear formulas define a witness \((\Gamma, \beta)\), which can be computed with an efficient linear program.
To obtain a rigorous proof, we must produce an exact witness.
That means solving the linear program with exact rational arithmetic --- an operation that is orders of magnitude slower than solving with floating-point numbers.
Instead of insisting on exact equality, we weaken the witness condition slightly, allowing for small perturbations.
These relaxed conditions can be solved efficiently in floating-point arithmetic with a numerical linear programming solver and then rationalized with the residual margin ensuring the conditions still hold.
\begin{theoremrep}[Zonotope containment] \label{thm:zonotope}
  The following is a derived proof rule of \dL:
    \begin{prooftree}
        \AX$\fCenter \vdash \matrixZonoB\matrixZonoB^+=I$
        \noLine
        \UI$\fCenter \vdash \vecZonoB - \vecZonoA = \matrixZonoB\beta$
        \noLine
        \UI$\fCenter \vdash \normInf[\matrixZonoB \Gamma - \matrixZonoA] \leq \varepsilon$
        \noLine
        \UI$\fCenter \vdash \normInf[(\Gamma, \beta)] \leq 1
        - \varepsilon\normInf[\matrixZonoB^+]$
        \UI$\fCenter \formulaZonoA(x) \vdash \formulaZonoB(x)$
      \end{prooftree}
\end{theoremrep}
\begin{proof}
  First, we fix an arbitrary \(\lambda\) with \(\normInf[\lambda] \leq 1\) for the zonotope in the antecedent. Next, we cut in two additional assumptions:
  \begin{prooftree}
    \AXC{\(\vdash \normInf[\matrixZonoA - \matrixZonoB\Gamma] \leq \epsilon\)}
    \LeftLabel{\(\R\)}
    \UIC{\(\vdash \exists R (\normInf[R] \leq \epsilon \land \matrixZonoA =
    \matrixZonoB \Gamma + R)\)}
    \AXC{\(\vdash \vecZonoB = \matrixZonoB \beta + \vecZonoA\)}
    \AXC{\dots}
    \LeftLabel{\ruleref{cut},\ruleref{cut}}
    \TIC{\(\normInf[\lambda] \leq 1\vdash \exists \mu:( \matrixZonoA\lambda +
    \vecZonoA = \matrixZonoB \mu + \vecZonoB \land \normInf[\mu]\leq 1)\)}
  \end{prooftree}

Then, we fix an arbitrary \(R\) with \(\normInf[R] \leq \epsilon\),
choose \(\mu = (\Gamma + \matrixZonoB^+R) \lambda -
\beta\) for some \(\Gamma \in \R\) and \(\beta \in \R\) and split the
target into two subgoals:
\begin{enumerate}
  \item Firstly, we show that the equality in A holds:
    \begin{prooftree}
      \AXC{*}
      \LeftLabel{\(\R\)}
      \UIC{\(\matrixZonoB\matrixZonoB^+=I\vdash (\matrixZonoB \Gamma + R)\lambda + \vecZonoA = \matrixZonoB
          ((\Gamma + \matrixZonoB^+R)
      \lambda -\beta) + \matrixZonoB \beta + \vecZonoA\)}
      \AXC{\(\matrixZonoB\matrixZonoB^+=I\)}
      \LeftLabel{\(\R\),\ruleref{cut}}
      \BIC{\(\vecZonoB = \matrixZonoB \beta + \vecZonoA, \matrixZonoA = \matrixZonoB \Gamma + R\vdash
          \matrixZonoA\lambda + \vecZonoA = \matrixZonoB ((\Gamma + \matrixZonoB^+R)
      \lambda -\beta) + \vecZonoB\)}
    \end{prooftree}
  \item Secondly, we show that the zonotope inequality holds:
    \begin{prooftree}
      \AXC{\(\vdash \normInf[(\Gamma,\beta)] \leq 1 -
      \epsilon \normInf[\matrixZonoB^+]\)}
      \LeftLabel{\(\R\)}
      \UIC{\(\normInf[\lambda] \leq 1, \normInf[R] \leq \epsilon
          \vdash \normInf[(\Gamma,\beta)]
          \normInf[\lambda] +
          \normInf[\matrixZonoB^+]\normInf[R]\normInf[\lambda] \leq 1\)}
      \UIC{\(\normInf[\lambda] \leq 1, \normInf[R] \leq \epsilon
      \vdash \normInf[(\Gamma,\beta)]
      \normInf[(\lambda,-1)] +
      \normInf[(\matrixZonoB^+R,0)]\normInf[(\lambda,-1)] +
  \normInf[\beta] \leq 1\)}
      \LeftLabel{\(\R\),\(\norm\)}
      \UIC{\(\normInf[\lambda] \leq 1, \normInf[R] \leq \epsilon
      \vdash \normInf[(\Gamma + \matrixZonoB^+R,\beta) (\lambda, - 1)] \leq 1\)}
      \UIC{\(\normInf[\lambda] \leq 1, \normInf[R] \leq \epsilon
      \vdash \normInf[(\Gamma + \matrixZonoB^+R) \lambda - \beta] \leq 1\)}
    \end{prooftree}
\end{enumerate}
\end{proof}
Equipped with \cref{thm:zonotope}, we have all the necessary ingredients to carry out an end-to-end verification of a concrete system.
\section{Evaluation}
In this section, we compute the control envelopes of two examples and demonstrate how to verify them using the theorem from \cref{sec:verification}.
The key challenge here is to handle the different representations that are used in reachability tools and the theorem prover KeYmaera X. 

To compute RCI sets for sampled-data systems, we use the approach by Schäfer et al. \cite{Schaefer2024a}.
The authors represent both the RCI and over-approximations of the reachable sets as zontopes.
Using the witness conditions \eqref{eq:sadraddini_zono_cont}, the one-step invariance, the one-step safety and control-admissibility in \cref{def:RCI} are formulated as constraints in an optimization problem returning an RCI set with maximum volume.
In order to verify the envelope in \dL's formalism, we need a post-processing step: we rationalize center vector and generator matrix of the output zonotopes\footnote{A suitable operator is implemented in the MATLAB function \texttt{rat}, see \url{https://de.mathworks.com/help/matlab/ref/rat.html}.}.
Then, we compute a provably correct \dL Taylor model using the approach implemented in KeYmaera~X\footnote{\url{https://github.com/LS-Lab/KeYmaeraX-release/blob/master/keymaerax-core/src/main/scala/org/keymaerax/btactics/TaylorModel.scala}}. Finally, we verify the zonotope containment visually.

\subsection{Double Integrator}

The dynamics of the double integrator are governed by the following differential equations \cite[Sec.~V.A]{Gruber2021_CSL}:
\begin{equation*} \label{eq:doubleIntegrator}
	\begin{split}
		& x'_1 = x_2 + w_1, \\
		& x'_2 = \frac{1}{m} u + w_2,
	\end{split}
\end{equation*}
where the system states are the position $x_1$ and the velocity $x_2$ of the point-mass, the system input is the force $u = F$, and the weight of the point-mass is $m = 1\kg$. 
The state constraints are $\abs{x_1} \leq 1\m$ and $\abs{x_2} \leq 1\ms$, the input constraint is $\abs{u}\leq 1 \Ne$, and the set of disturbances is $w_1 \in [-0.1,0.1]\ms$ and $w_2 \in [-0.1,0.1]\mss$. 
The sampling time is $\constSamplingTime=\unit[0.1]{s}$.
In \dL, we use a slightly enlarged initial condition given by the formula
\[X_0(x,u) \equiv \exists \lambda \left(x_1 = \frac{11}{10}\lambda_1 \land x_2 = \frac{11}{10}\lambda_2 \land u = \frac{11}{10}\lambda_3 \land \normInf[\lambda] \leq 1\right).\]
By Picard iteration we obtain the following polynomial approximation
\begin{align*}
  p_{x_1}(t, \lambda) &= \lambda_1 + t\lambda_2 + \frac{t^2}{2}\lambda_3,\\
  p_{x_2}(t, \lambda) &= \lambda_2 + t\lambda_3,\\
  p_{u}(t, \lambda) &= \lambda_3.
\end{align*}
With interval arithmetic we then obtain the following provable error bounds:
\begin{align*}
  \opIntLower{\intA}_{x_1}(t) = -101020\cdot 10^{-11}t &\quad\opIntUpper{\intA}_{x_1}(t) = 101020\cdot 10^{-11}t,\\
  \opIntLower{\intA}_{x_2}(t) = -10^{-6}t &\quad \opIntUpper{\intA}_{x_2}(t) = 10^{-6}t,\\
  \opIntLower{\intA}_{u}(t) = -10^{-6}t &\quad  \opIntUpper{\intA}_{u}(t) = 10^{-6}t.
\end{align*}
Together, the polynomial \(p\) and the error interval \(\intA(t)\) yield the Taylor model \((p,\intA)\) on \([0,\constSamplingTime]\).
We numerically compute an RCI set and, using the Taylor model, check the premises of \cref{thm:reach_discrete} and \cref{thm:reach_interval} (See \cref{fig:double_integrator_rci}).
Note that, although we visually verify containment here, a formal proof can be derived directly by applying \cref{thm:zonotope}.
\subsection{Controlled Jet Engine}

Next, we consider the Moore-Greitzer model of a jet engine \cite{Krstic1995} whose dynamics are governed by
\begin{align} \label{eq:jetEngine}
	\begin{split}
		x'_1 &= -x_2 - \frac{3}{2} x_1^2 - \frac{1}{2} x_1^3 + w, \\
		x'_2 &= u,
	\end{split}
\end{align}
where the system states are the mass flow $x_1$ and the pressure rise $x_2$.
The state constraints are $\abs{x_1} \leq 0.2$ and $\abs{x_2} \leq 0.2$, the input constraint is $\abs{u}\leq 0.3$, and the set of disturbances is $w \in [-0.025,0.0.025]$. 
Measurements are taken with a sampling time of $\constSamplingTime = 0.1$ time units. Computing the associated \dL Taylor model with conservative initial conditions
\[X_0(x,u) \equiv \exists \lambda \left(x_1 = \frac{3}{10}\lambda_1 \land x_2 = \frac{3}{10}\lambda_2 \land u = \frac{3}{10}\lambda_3 \land \normInf[\lambda] \leq 1\right)\]
yields the following provable polynomial approximation
\begin{align*}
  p_{x_1}(t, \lambda) &= \lambda_1 - t\lambda_2 + \frac{3t}{2}\lambda_1^2 - \frac{t^2}{2}\lambda_3 - \frac{3t^2}{2}\lambda_1\lambda_2 - \frac{t}{2}\lambda_1^3\\
  p_{x_2}(t, \lambda) &= \lambda_2 + t\lambda_3\\
  p_{u}(t, \lambda) &= \lambda_3
\end{align*}
with error bounds
\begin{align*}
  \opIntLower{\intA}_{x_1}(t) = -28605705206\cdot 10^{-11}t &\quad\opIntUpper{\intA}_{x_1}(t) = 27585076206\cdot 10^{-11}t,\\
  \opIntLower{\intA}_{x_2}(t) = -10^{-6}t &\quad \opIntUpper{\intA}_{x_2}(t) = 10^{-6}t,\\
  \opIntLower{\intA}_{u}(t) = -10^{-6}t &\quad  \opIntUpper{\intA}_{u}(t) = 10^{-6}t.
\end{align*}
Again, using the Taylor model (\(p,\intA\)), we can compute the zonotopes form \cref{thm:reach_discrete} and \cref{thm:reach_interval} (See \cref{fig:jet_engine_rci}).
\begin{figure}
  \centering
  \begin{subfigure}[b]{0.48\linewidth}
    \centering
    \includegraphics[width=\linewidth]{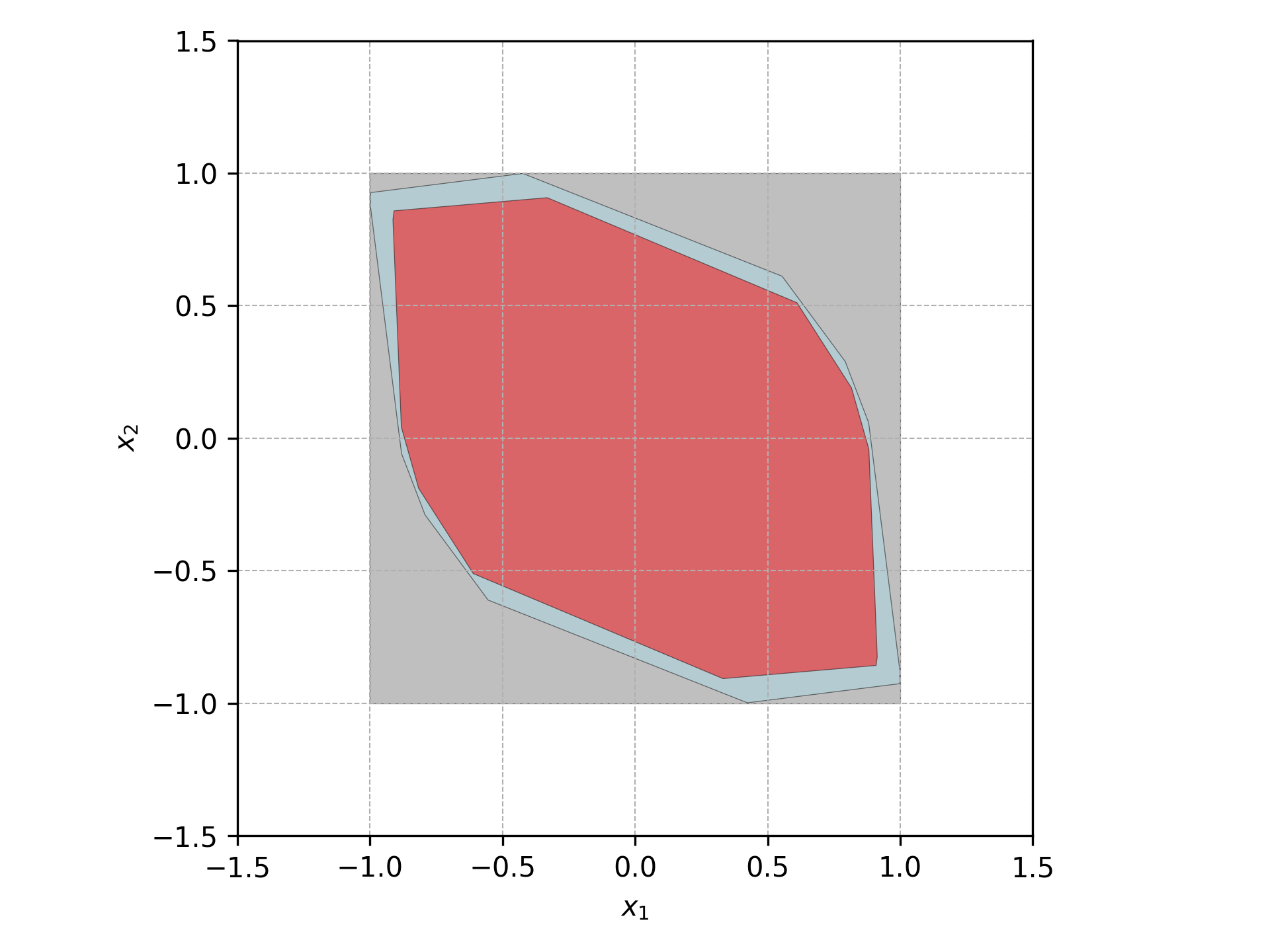}
    \caption{Double integrator}
    \label{fig:double_integrator_rci}
  \end{subfigure}
  \hfill
  \begin{subfigure}[b]{0.48\linewidth}
    \centering
    \includegraphics[width=\linewidth]{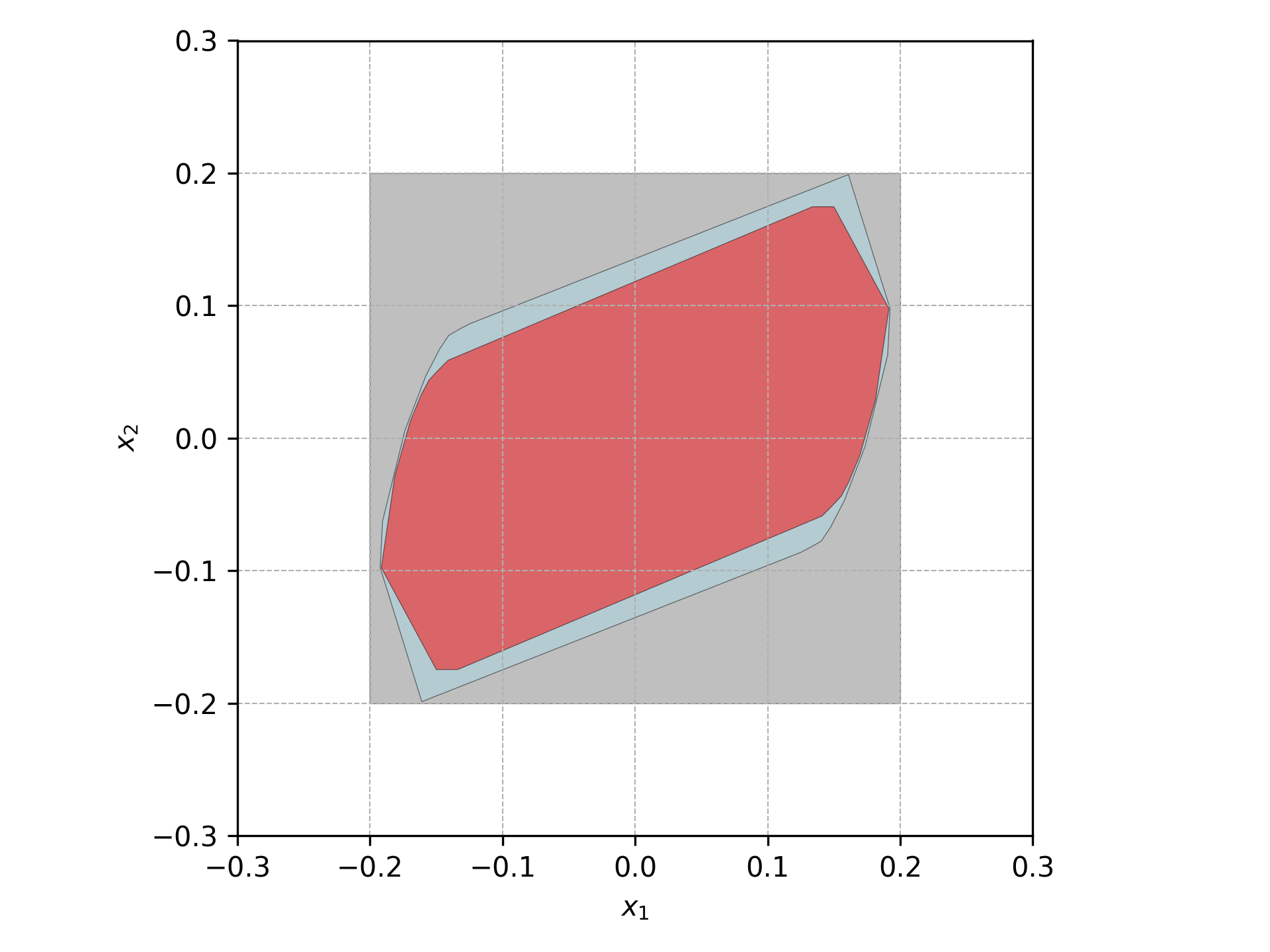}
    \caption{Jet engine}
    \label{fig:jet_engine_rci}
  \end{subfigure}
  \caption{%
    Comparison of the numerically computed robust control-invariant set (blue), its Taylor-model
    reachability over-approximation after one sampling period (red), and the safety region (gray).
  }
  \label{fig:two-panel}
\end{figure}
\section{Related Work}
\paragraph{Reachability analysis.} Reachability analysis itself can be categorized into several
techniques: simulation-based techniques \cite{Girard2006b},
Hamilton-Jacobi techniques \cite{Chen2018}, and set propagation
techniques \cite{Althoff2021b}. The disadvantage of simulation-based
techniques and Hamilton-Jacobi techniques is that they scale
exponentially in the number of continuous state variables, while many
set propagation techniques scale polynomially \cite{Althoff2021b}.
Because the main purpose of including reachability analysis into
theorem proving is to deduce properties of complex dynamics with
potentially many continuous state variables, we will focus on set
propagation techniques. This technique is also currently
predominantly used in the International Competition on Verifying
Continuous and Hybrid Systems \cite{Althoff2024a,Geretti2024}.
Further advantages of set-based reachability analysis are that it can
be fully automated \cite{Wetzlinger2021,Wetzlinger2023a} and easily
interpreted due to its resemblance with the numerical simulation of systems.
\paragraph{Computing Invariant Sets Using Reachability Analysis}
Computing invariant sets has a rich history in control theory due to
their many applications: they serve as terminal regions in model
predictive control \cite{Rawlings2022} or are employed as part of
supervisory safety-filters in, e.g., learning-based control
\cite{Mitchell2016,Wabersich2018}.
Since we aim at reducing conservatism, we focus on the computation of
the largest invariant set, also known as the maximal invariant set.
The maximal invariant set can be obtained using the standard set
recursion introduced in \cite{Bertsekas1972,Rungger2017}.
However, this procedure usually fails to terminate in finite time and
the computational complexity of the required set operations restricts
the applicability to low-dimensional systems.

The latter also holds for algorithms for approximating the maximal
invariant set by gridding the sate space \cite{Bravo2005,Brown2023}.
Thus, most approaches in the literature formulate an optimization
problem to compute a possibly large invariant set.

The most popular set representations of the invariant set are
ellipsoids \cite{Chen2003,Lazar2018} and polytopes
\cite{Rakovic2010_TAC,Gupta2019,Fiacchini2010_Automatica,BenSassi2012_SCL}.
Due to their low representation complexity, algorithms using
ellipsoids as the set representation scale better to
higher-dimensional systems at the cost of more conservative results.
On the other hand, polytopic invariant sets enable more flexibility
and, thus, larger invariant sets while sacrificing computational efficiency.

Level sets are an even more flexible representation of invariant sets
that can be computed using Hamilton-Jacobi reachability analysis
\cite{Xue2019,Xue2021} and control-barrier functions
\cite{Ames2019,Rauscher2016,Xu2015}.
To circumvent the exponential complexity of solving the associated
partial differential equation numerically (Hamilton-Jacobi
reachability analysis) or synthesizing the barrier certificate from
simulations \cite{Ames2019}, the computation of an invariant set can
be relaxed into a (sequence of) semidefinite program(s) using
sum-of-squares programming.
However, an invariance-enforcing controller must be designed prior to
computing the invariant set \cite{Xue2019,Xue2021} or the approach
suffers from poor scalability due to the large number of variables of
the semi-definite program \cite{Korda2014,Ahmadi2017}.

The approaches reviewed above typically either consider discrete-time
systems, e.g., \cite{Fiacchini2010_Automatica,Lazar2018} and, thus,
do not check constraint satisfaction in between sampling times, or
enforce invariance for the continous-time dynamical system, e.g.,
\cite{Ames2019,BenSassi2012_SCL}, which is unnecessarily conservative.
As an alternative, invariant sets for sampled-data systems have been
characterized in \cite{Rakovic2017}: sampled-data systems are
continuous-time systems that are controlled by a digital controller;
similary, measurements are only taken at discrete points in time
\cite{Mitchell2016}.
Crucially, the system can leave the invariant set in between sampling
times, which reduces conservatism.
This notion of invariance has been employed in \cite{Gruber2021a} to
introduce so-called safe sets of linear systems.
Since this approach represents the invariant set as a zonotope,
invariant sets of high-dimensional systems can be computed efficiently.
This concept has been extended to ellipsoidal sets in
\cite{Kulmburg2024_TAC} and nonlinear systems in
\cite{Schaefer2024b,Schaefer2024a}.
\paragraph{Deductive verification of control problems} 
Differential dynamic logic (\dL) has been successfully applied in several control domains, including air traffic control, train control, and ground robots to formally prove safety properties \cite{bohrerFormalSafetyNet2019,jeanninFormalVerificationACAS2015,kabraVerifiedTrainControllers2022a}.
It has also been employed for the deductive verification of control system stability \cite{tanDeductiveStabilityProofs2021,tanVerifyingSwitchedSystem2022}.
The control-system meta-model has been extended to incorporate the environment, with a focus on identifying conditions that prevent proofs of safety from being invalidated by modeling errors \cite{selvarajProvingThatUnsafe2024}.
In contrast, our work focuses on a simplified controller-plant model and formalizes in \dL the verification of synthesized controllers by reducing closed-loop analysis to continuous-time safety and discrete-time invariance.

The problem of control-envelope synthesis has been studied in the context of~\dL before \cite{arechigaUsingVerifiedControl2014,kabraCESARControlEnvelope2024}.
Both of these works approach the problem from a logical perspective and do not leverage existing techniques and tools developed in the field of reachability analysis field.
By comparison, our work integrates these two viewpoints by combining established numerical methods.

\section{Conclusion}
In this paper, we established a link between two traditionally separate research fields: reachability analysis and theorem proving.
We showed how zonotope-based reachable-set computations can be encoded in the \dL formalism and how a control envelope can be formally verified.
Although the differing levels of representation between these tools posed nontrivial technical challenges, our case studies demonstrate that these obstacles can be overcome.
By combining the computational efficiency of reachability analysis with the deductive rigor of theorem proving, we achieve a verification workflow that is both scalable and formally sound.
This work represents only the first step toward a more unified formal-methods ecosystem.
In future work, we plan to explore how reachability methods can be even more tightly integrated, reducing the boundaries between research fields.
\section*{Acknowledgements}
This work was supported  by the German Research Foundation - SFB 1608 - under grant number 501798263 and by an Alexander von Humbolt Professorship.
Additionally, we would like to thank Fabian Immler for his valuable Taylor model implementation in KeYmaera~X.
\renewcommand{\doi}[1]{doi: \href{https://doi.org/#1}{\nolinkurl{#1}}}
\bibliographystyle{splncs04}
\bibliography{reachability, althoff_own, althoff_other, lukas}

\newpage
\appendix
\section{Proof rules}
\ruleset{impliesLeft}{\(\rightarrow\mathrm{L}\)}
\ruleset{andLeft}{\(\land \mathrm{L}\)}
\ruleset{orRight}{\(\lor\)R}
\ruleset{id}{\textup{id}}
\ruleset{weakeningLeft}{\textup{WL}}
\ruleset{forallRight}{\(\forall\)\textup{R}}
\ruleset{existsLeft}{\(\exists\)\textup{L}}
\ruleset{existsRight}{\(\exists\)\textup{R}}
\ruleset{cut}{\textup{cut}}
\ruleset{loop}{\textup{loop}}
\ruleset{g}{\textup{G}}
\ruleset{monotonicity}{\textup{MR}}
\ruleset{existsGeneralization}{\(\exists\)\textup{G}}

\begin{theorem}[Propositional rules]
  The following are sound proof rules of \dL:
  \begin{mathpar}
    \ruletag{impliesRight}
    \AXC{\(\Gamma, P \vdash Q, \Delta\)}
    \UIC{\(\Gamma \vdash P \rightarrow Q, \Delta\)}
    \DisplayProof
    \and
    \ruletag{impliesLeft}
    \AXC{\(\Gamma \vdash P, \Delta\)}
    \AXC{\(\Gamma, Q \vdash \Delta\)}
    \BIC{\(\Gamma, P \rightarrow Q \vdash \Delta\)}
    \DisplayProof
    \and
    \ruletag{andLeft}
    \AXC{\(\Gamma, P, Q \vdash \Delta\)}
    \UIC{\(\Gamma, P \land Q \vdash \Delta\)}
    \DisplayProof
    \and
    \ruletag{andRight}
    \AXC{\(\Gamma \vdash P, \Delta\)}
    \AXC{\(\Gamma \vdash Q, \Delta\)}
    \BIC{\(\Gamma \vdash P \land Q, \Delta\)}
    \DisplayProof
    \and
    \ruletag{orRight}
    \AXC{\(\Gamma \vdash P, Q, \Delta\)}
    \UIC{\(\Gamma \vdash P \lor Q, \Delta\)}
    \DisplayProof
    \and
    \ruletag{id}
    \AXC{*}
    \UIC{\(\Gamma, P \vdash P, \Delta\)}
    \DisplayProof
    \and
    \ruletag{weakeningLeft}
    \AXC{\(\Gamma \vdash \Delta\)}
    \UIC{\(\Gamma, P \vdash \Delta\)}
    \DisplayProof
    \and
    \ruletag{forallRight}
    \AXC{\(\Gamma \vdash p(y), \Delta\)}
    \UIC{\(\Gamma \vdash \forall x p(x), \Delta\)}
    \DisplayProof

    \ruletag{existsLeft}
    \AXC{\(\Gamma, p(y) \vdash \Delta\)}
    \UIC{\(\Gamma, \exists x p(x) \vdash \Delta\)}
    \DisplayProof

    \ruletag{existsRight}
    \AXC{\(\Gamma \vdash p(e), \Delta\)}
    \UIC{\(\Gamma \vdash \exists x p(x), \Delta\)}
    \DisplayProof

    \and
    \ruletag{cut}
    \AXC{\(\Gamma \vdash C, \Delta\)}
    \AXC{\(\Gamma, C \vdash \Delta\)}
    \BIC{\(\Gamma \vdash \Delta\)}
    \DisplayProof
  \end{mathpar}

  \begin{mathpar}
    \ruletag{loop}
    \AXC{\(\Gamma \vdash J, \Delta\)}
    \AXC{\(J \vdash [a]J\)}
    \AXC{\(J \vdash P\)}
    \TIC{\(\Gamma \vdash [a^*]P, \Delta\)}
    \DisplayProof
    \and
    \ruletag{g}
    \AxiomC{\(\vdash P\)}
    \UnaryInfC{\(\Gamma \vdash [a]P, \Delta\)}
    \DisplayProof
    \and
    \ruletag{monotonicity}
    \AXC{\(\Gamma \vdash [a]Q, \Delta\)}
    \AXC{\(Q\vdash P\)}
    \BIC{\(\Gamma \vdash [a]P, \Delta\)}
    \DisplayProof

    \ruletag{existsGeneralization}
    \AXC{\(\Gamma, \exists y p(x,y) \vdash \Delta\)}
    \UIC{\(\Gamma, p(x,y) \vdash \Delta\)}
    \DisplayProof
  \end{mathpar}
\end{theorem}

\section{Axioms}

\axset{normTriangle}{$\norm_{\leq}$}
\axset{normHomogenity}{$\norm_{=}$}
\axset{normOp}{$\norm_{\leq A}$}
\axset{normDef}{$\normInf[\,\cdot\,]$}
\begin{theorem}[Arithmetic axioms]
  The following are derived axioms of \dL:
  \begin{align*}
    &\blue{\normInf[x + y]} \leq \normInf[x] + \normInf[y]
    \axtag{normTriangle}\\
    &\blue{\normInf[cx]} = \abs{\,c\,} \cdot \normInf[x] \axtag{normHomogenity}\\
    &\blue{\normInf[Ax]} \leq \normInf[A] \cdot \normInf[x]
    \axtag{normOp}\\
    &\blue{\normInf[x] \leq 1} \leftrightarrow -1 \leq x \leq 1 \axtag{normDef}
  \end{align*}
\end{theorem}

\axset{boxTest}{\([?]\)}
\axset{boxNonDetAssign}{\([\define\!*]\)}
\axset{boxDX}{\textup{DX}}
\axset{boxDW}{\textup{DW}}
\axset{boxOdeIdm}{\textup{DID}}
\axset{boxAnd}{\([\land]\)}
\begin{theorem}[Box axioms \cite{platzerCompleteUniformSubstitution2017b}]
  The following are sound axioms of \dL:
  \begin{align*}
    &\blue{[?Q]P} \leftrightarrow (Q \rightarrow P) \axtag{boxTest} \\
    &\blue{[x:=*]P} \leftrightarrow \forall x P
    \axtag{boxNonDetAssign} \\
    &\blue{[x'=f(x) \AND Q]P(x)} \rightarrow (Q \rightarrow P(x)) \axtag{boxDX} \\
    &\blue{[x'=f(x) \AND Q]P} \leftrightarrow [x'=f(x) \AND Q](Q
    \rightarrow P) \axtag{boxDW} \\
    &\blue{[x'=f(x) \AND Q][x'=f(x) \AND Q]P} \leftrightarrow [x'=f(x) \AND Q]P \axtag{boxOdeIdm} \\
    &\blue{[a](P\land Q)} \leftrightarrow [a]P \land [a]Q \axtag{boxAnd}
  \end{align*}
\end{theorem}

The following sound proof rule is used to establish validity of Taylor models. 

\begin{theorem}[Closed Differential Induction]
  The following proof rule is sound:
  \ruleset{dIC}{$\overline{\mathrm{dI}}$}
  \begin{mathpar}
    \ruletag{dIC}
    \AXC{\(\Gamma \vdash e \geq 0\)}
    \AXC{\(\Gamma \vdash [x' = f(x) \& Q \land e \geq 0]e' > 0\)}
    \BIC{\(\Gamma \vdash [x' = f(x) \& Q]e \geq 0\)}
    \DisplayProof
  \end{mathpar}
\end{theorem}

\section{Deferred Proofs}

\begin{lemma}[Discrete time control] \label{lemma:equivalence}
Let \(P\) be an arbitrary formula. Then, the following sequent is
provable in \dL:
\begin{align*}
  0 \leq t < \Delta t \vdash [\sys]P \leftrightarrow [\plant]P.
\end{align*}
\end{lemma}
\begin{proof}
First, we show that for \(0 \leq t < \Delta t\) running the
controller has not affect on the post condition:
\begin{prooftree}
  \AXC{*}
  \LeftLabel{\ruleref{id}, \(\leftrightarrow\)}
  \UIC{\(0 \leq t < \Delta t\vdash P  \leftrightarrow P \)}
  \LeftLabel{\(\R\)}
  \UIC{\(0 \leq t < \Delta t\vdash ((t=\Delta t \rightarrow [u:=*;
      ?\formulaEnvelope(x, u); t:= 0]P) \land (t \neq \Delta t \rightarrow P))
  \leftrightarrow P \)}
  \LeftLabel{\axref{boxTest},\axref{boxChoice}}
  \UIC{\(0 \leq t < \Delta t\vdash [?t = \Delta t; u:=*; ?\formulaEnvelope(x, u);
        t:= 0 \; \cup \; ? t
  \neq \Delta t]P \leftrightarrow P \)}
  \UIC{\(0 \leq t < \Delta t\vdash [\ctrl]P \leftrightarrow P \)}
\end{prooftree}
From the argument above, we can conclude that the whole system is
equivalent to just running the plant for \(0 \leq t < \Delta t\):
\begin{prooftree}
  \AXC{*}
  \UIC{\(0 \leq t < \Delta t \vdash [\ctrl][\plant]P \leftrightarrow
  [\plant]P\)}
  \LeftLabel{\axref{boxSeq}}
  \UIC{\(0 \leq t < \Delta t \vdash [\sys]P \leftrightarrow [\plant]P\)}
\end{prooftree}
\end{proof}
\begin{lemma}[Loop invariant] \label{lemma:loop}
The formula
\begin{align*}
  J &\equiv [\plant]X(x) \land [\plant](t = \Delta t \rightarrow
  \exists u : S(x,u)) \land 0 \leq t \leq \Delta t
\end{align*}
is a loop invariant for the safety property \(S(x,u), t = 0 \vdash
[\sys^*]X(x).\)
This means
\begin{enumerate}
  \item \(S(x,u), t = 0 \vdash J,\)
  \item \(J \vdash [\sys]J,\)
  \item \(J \vdash X(x)\)
\end{enumerate}
are provable is \dL.
\end{lemma}
\begin{proof}
We show that all three sequents are provable in \dL:
\begin{enumerate}
  \item The proof of the initial condition is trivial by the
    defintition of robust control invariant set:
    \begin{prooftree}
      \AXC{\(S(x,u),t=0\vdash [\plant]X(x)\)}
      \AXC{\(S(x,u),t=0\vdash [\plant]S_{\Delta t}(t,x)\)}
      \AXC{*}
      \LeftLabel{\(\R\)}
      \UIC{\(S(x,u),t=0\vdash \leq t \leq \Delta t\)}
      \LeftLabel{\ruleref{andRight}}
      \TIC{\(S(x,u),t=0\vdash J\)}
    \end{prooftree}
  \item We start by showing the discrete invariance property.
    We split the assumption \(0 \leq t \leq \Delta t\) into two
    cases. First, consider the case where the controller gets
    triggered, i.e., \(t = \Delta t\):
    \begin{prooftree}
      \AXC{\ref{item:invariance}\(\equiv t = 0, S(x,u) \vdash
      [\plant]S_{\Delta t}(t,x)\)}
      \LeftLabel{\ruleref{impliesRight},\ruleref{forallRight},\ruleref{weakeningLeft}}
      \UIC{\(t = \Delta t \vdash \forall u : (S(x,u) \rightarrow [t:=
      0][\plant]S_{\Delta t}(t,x))\)}
      \LeftLabel{\axref{boxTest},\axref{boxNonDetAssign},\axref{boxTest}}
      \UIC{\(t = \Delta t \vdash [?t = \Delta t; u:=*; ?S(x, u); t:=
      0][\plant]S_{\Delta t}(t,x)\)}
      \LeftLabel{\(\R\),\ruleref{andRight},\axref{boxChoice}}
      \UIC{\(t = \Delta t \vdash [\ctrl][\plant]S_{\Delta t}(t,x)\)}
      \LeftLabel{\axref{boxOdeIdm}}
      \UIC{\(t = \Delta t \vdash [\ctrl][\plant][\plant]S_{\Delta t}(t,x)\)}
      \LeftLabel{\ruleref{weakeningLeft},\axref{boxSeq}}
      \UIC{\(J, t = \Delta t \vdash [\sys][\plant]S_{\Delta t}(t,x)\)}
    \end{prooftree}
    Next, we consider the case with \(0\leq t < \Delta t\):
    \begin{prooftree}
      \AXC{*}
      \LeftLabel{\ruleref{id},\ruleref{andLeft}}
      \UIC{\(J, 0 \leq t < \Delta t \vdash [\plant]S_{\Delta t}(t,x)\)}
      \LeftLabel{Lemma \ref{lemma:equivalence}}
      \UIC{\(J, 0 \leq t < \Delta t \vdash [\ctrl][\plant]S_{\Delta t}(t,x)\)}
      \LeftLabel{\axref{boxOdeIdm}}
      \UIC{\(J, 0 \leq t < \Delta t \vdash
      [\ctrl][\plant][\plant]S_{\Delta t}(t,x)\)}
      \LeftLabel{\axref{boxSeq}}
      \UIC{\(J, 0 \leq t < \Delta t \vdash [\sys][\plant]S_{\Delta t}(t,x)\)}
    \end{prooftree}

    Next, we prove the safety property in a similar manner starting
    with the \(t=\Delta t\) case:
    \begin{prooftree}
      \AXC{\ref{item:safety}\(\equiv t = 0, S(x,u) \vdash [\plant]X(x)\)}
      \LeftLabel{\ruleref{impliesRight},\ruleref{forallRight},\ruleref{weakeningLeft}}
      \UIC{\(t = \Delta t \vdash \forall u : (S(x,u) \rightarrow [t:=
      0][\plant]X(x))\)}
      \LeftLabel{\axref{boxTest},\axref{boxNonDetAssign},\axref{boxTest}}
      \UIC{\(t = \Delta t \vdash [?t = \Delta t; u:=*; ?S(x, u); t:=
      0][\plant]X(x)\)}
      \LeftLabel{\(\R\),\ruleref{andRight},\axref{boxChoice}}
      \UIC{\(t = \Delta t \vdash [\ctrl][\plant]X(x)\)}
      \LeftLabel{\ruleref{weakeningLeft},\axref{boxOdeIdm}}
      \UIC{\(J, t = \Delta t \vdash [\ctrl][\plant][\plant]X(x)\)}
      \LeftLabel{\axref{boxSeq}}
      \UIC{\(J, t = \Delta t \vdash [\sys][\plant]X(x)\)}
    \end{prooftree}
    Finally, the other case follow from:
    \begin{prooftree}
      \AXC{*}
      \LeftLabel{\ruleref{id},\ruleref{andLeft}}
      \UIC{\(J, 0 \leq t < \Delta t \vdash [\plant]X(x)\)}
      \LeftLabel{Lemma \ref{lemma:equivalence}}
      \UIC{\(J, 0 \leq t < \Delta t \vdash [\ctrl][\plant]X(x)\)}
      \LeftLabel{\axref{boxOdeIdm}}
      \UIC{\(J, 0 \leq t < \Delta t \vdash [\ctrl][\plant][\plant]X(x)\)}
      \LeftLabel{\axref{boxSeq}}
      \UIC{\(J, 0 \leq t < \Delta t \vdash [\sys][\plant]X(x)\)}
    \end{prooftree}
  \item The safety property follows from the \axref{boxDX} axiom:
    \begin{prooftree}
      \AXC{*}
      \LeftLabel{\ruleref{id}}
      \UIC{\(0 \leq t \leq \Delta t,[\plant]X(x), t \leq \Delta t
      \rightarrow X(x) \vdash X(x)\)}
      \LeftLabel{\ruleref{id},\axref{boxDX},\ruleref{cut}}
      \UIC{\(0 \leq t \leq T, [\plant]X(x) \vdash X(x)\)}
    \end{prooftree}
\end{enumerate}
\end{proof}
\end{document}